\documentclass[preprint,12pt]{elsarticle}
\usepackage{lineno,natbib,hyperref,eurosym,amsmath,amsfonts,amssymb,setspace,graphicx,amsthm}

\newtheorem{theorem}{Theorem}

\newtheorem{corollary}[theorem]{Corollary}

\newtheorem{example}[theorem]{Example}
\newtheorem{lemma}[theorem]{Lemma}

\begin{document}

\begin{frontmatter}


\ead{mkoroglu@yildiz.edu.tr, msari@yildiz.edu.tr}

\title{On MDS Negacyclic LCD Codes}


\author{Mehmet E. Koroglu and Mustafa Sar\i}

\address{Department of Mathematics, Faculty of Art and Science\\
Y\i ld\i z Technical University, Esenler, Istanbul 34220, Turkey}

\begin{abstract}
 Linear codes with complementary duals (LCD) have a great deal of significance amongst linear codes. Maximum distance separable (MDS) codes are also an important class of linear codes since they achieve the greatest error correcting and detecting capabilities for fixed length and dimension. The construction of linear codes that are both LCD and MDS is a hard task in coding theory. In this paper, we study the constructions of LCD codes that are MDS from negacyclic codes over finite fields of odd prime power $q$ elements.
We construct four families of MDS negacyclic LCD codes of length $n|\frac{{q-1}}{2}$, $n|\frac{{q+1}}{2}$ and a family of negacyclic LCD codes of length $n=q-1$. Furthermore, we obtain five families of $q^{2}$-ary Hermitian MDS negacyclic LCD codes of length $n|\left( q-1\right)$ and four families of Hermitian negacyclic LCD codes of length $n=q^{2}+1.$ For both Euclidean and Hermitian cases the dimensions of these codes are determined and for some classes the minimum distances are settled. For the other cases, by studying $q$ and $q^{2}$-cyclotomic classes we give lower bounds on the minimum distance.
\end{abstract}

\end{frontmatter}

\section{Introduction}

Linear codes with complementary-duals (LCD codes), which was introduced by
Massey in 1992 (see \cite{Massey}), have many applications in cryptography,
communication systems, data storage and consumer electronics. A linear code
is called an LCD code if $\mathcal{C}^{\bot }\cap \mathcal{C}=\left\{
0\right\} .$ LCD codes provide an optimum linear coding solution for binary
adder channel \cite{Massey}, and in \cite{Massey1} it has been shown that
asymptotically good LCD codes exist. Further, in \cite{Sendrier} Sendrier
proved that LCD codes meet Gilbert-Varshamov bound. In \cite{Yang}, Yang and
Massey gave a necessary and sufficient condition for a cyclic code to have a
complementary dual. All LCD constacyclic codes of length $2l^{t}p^{s}$ has
been determined in \cite{Chen}. The LCD condition for a certain class of
quasi cyclic codes has been studied in \cite{Esmaeili}. In \cite{Dougherty},
Dougherty et al. have been given a linear programming bound on the largest
size of an LCD code of given length and minimum distance. Guneri et al.
introduced Hermitian LCD codes in \cite{Guneri}. In \cite{Zhu}, a class of
MDS negacyclic LCD codes of even length $\left. n\right\vert q-1$ has been
given. Carlet and Guiley have studied an application of LCD codes against
side-channel attacks, and have presented particular constructions for LCD
codes in \cite{Carlet}. MDS LCD codes over finite field $\mathbb{F}_{q}$
with even $q$ have been completely solved in \cite{Jin}. In \cite{Li}, Li et
al. have explored two special families of LCD cyclic codes, which are both
BCH codes. The authors of \cite{Li1} have constructed several families of
reversible cyclic codes over finite fields and have analyzed their
parameters. Galvez et al, gave exact values of dimension $k$ and length $n$
of a binary LCD code, where $1\leq k\leq n\leq 12.$ In \cite{Li2}, Li has
constructed some non MDS cyclic Hermitian LCD codes over finite fields and
has analysed their parameters. In \cite{Chen1}, Chen and Liu have proposed a different approach to obtain new LCD MDS codes from generalized Reed-Solomon
codes. In \cite{Beelen}, Beelen and Jin gave an explicit construction of
several classes of LCD MDS codes, using tools from algebraic function
fields. In \cite{Carlet1}, Carlet et al. have studied several constructions
of new Euclidean and Hermitian LCD MDS codes. In \cite{Sok}, Sok et al. have proved existence of optimal LCD codes over large finite fields and they have also gave methods to generate orthogonal matrices over finite fields and then apply them to construct LCD codes.

In this paper, we obtain four families of MDS negacyclic LCD codes and a
family of negacyclic LCD codes as follows:

\begin{enumerate}
\item For even $n>2,$ ${\left[ {n,n-2\lambda ,2\lambda +1}\right] _{q}},$
where $1\leq \lambda \leq \frac{{n-2}}{2},$ $q$ is an odd prime power and $n|%
\frac{{q-1}}{2}$ such that $n\neq 1.$

\item For odd $n,$ ${\left[ {n,n-2\lambda -1,2\left( {\lambda +1}\right) }%
\right] _{q}},$ where $0\leq \lambda \leq \frac{{n-3}}{2},$ $q$ is an odd
prime power and $n|\frac{{q-1}}{2}$ such that $n\neq 1.$

\item For even $n>2,$ ${\left[ {n,2\lambda,n-2\lambda+1}\right] _{q}}$,
where $1\leq \lambda\leq \frac{{n}}{2}-1,$ $q$ is an odd prime power and $n|%
\frac{{q+1}}{2}$ such that $n\neq 1.$

\item For odd $n,$ ${\left[ {n,2\lambda ,n-2\lambda +1}\right] _{q}},$ where
$1\leq \lambda \leq \frac{{n-1}}{2},$ $q$ is an odd prime power and $n|\frac{%
{q+1}}{2}$ such that $n\neq 1$.

\item ${\left[ {q + 1,4\lambda,d \ge \frac{{q + 3}}{2} - 2\lambda} \right]_q}%
,$ where $1 \le \lambda \le \frac{{q - 3}}{4},$ $q$ is an odd prime power
and $n=q+1$ such that $4|n.$
\end{enumerate}

We also obtain five families of negacyclic MDS Hermitian LCD codes and four
families of negacyclic Hermitian LCD codes as follows:

\begin{enumerate}
\item ${\left[ n,n-2l-2,2l+3\right] _{q^{2}}},$ where $2\nmid \gamma ,$ $%
0\leq l\leq \frac{q-4\gamma -1}{4\gamma },$ $q\equiv 1$ $\left( mod\text{ }%
4\right) ,$ and $n=\frac{{q-1}}{\gamma }>2.$

\item ${\left[ n,n-2l-2,2l+3\right] _{q^{2}}},$ where $2|\gamma ,$ $0\leq
l\leq \frac{q-4\gamma -1}{2\gamma },$ $q\equiv 1$ $\left( mod\text{ }%
4\right) ,$ $n=\frac{{q-1}}{\gamma }>2,$ and $n$ is even.

\item ${\left[ n,n-2l-1,2l+2\right] _{q^{2}}},$ where $2|\gamma ,$ $0\leq
l\leq \frac{q-3\gamma -1}{2\gamma },$ $q\equiv 1$ $\left( mod\text{ }%
4\right) ,$ $n=\frac{{q-1}}{\gamma }>2,$ and $n$ is odd.

\item ${\left[ n{,n-}\left( {\frac{{q-1}}{2}+2l+2}\right) {,d\geq \frac{{q-1}%
}{2{\gamma }}+l+2}\right] _{q^{2}}},$ where $2\nmid \gamma ,$ $0\leq l\leq
\frac{q-8\gamma -1}{4\gamma },$ $n=\frac{{q-1}}{\gamma }>4,$ and $q\equiv 1$
$\left( mod\text{ }4\right) .$

\item ${\left[ n,n-2l-1,2l+2\right] _{q^{2}}},$ where $2\nmid \gamma ,$ $%
0\leq l\leq \frac{q-2\gamma -1}{4\gamma },$ $q\equiv 3$ $\left( mod\text{ }%
4\right) ,$ and $n=\frac{{q-1}}{\gamma }>2.$

\item ${\left[ n,n-2l-1,2l+2\right] _{q^{2}}},$ where $2\left\vert \gamma
\right. ,$ $0\leq l\leq \frac{q-3\gamma -1}{2\gamma },$ $q\equiv 3$ $\left(
mod\text{ }4\right) ,$ $n=\frac{{q-1}}{\gamma }>2.$

\item ${\left[ n{,n-}\left( {\frac{{q-1}}{2}+2l+2}\right) {,d\geq \frac{{q-1}%
}{2}+l+2}\right] _{q^{2}}},$ where $2\nmid \gamma ,$ $0\leq l\leq \frac{%
q-6\gamma -1}{4\gamma },$ $n=\frac{{q-1}}{\gamma }>4,$ and $q\equiv 3$ $%
\left( mod\text{ }4\right) .$

\item ${\left[ q^{2}+1{,4l,d\geq }\frac{{q}^{2}-4l+3}{2}\right] _{q^{2}}},$
where $q$ is an odd prime power such that $q\equiv 1$ $(mod$ $4)$ and $\frac{%
\left( q-1\right) ^{2}}{4}\leq l\leq \frac{{q}^{2}{-1}}{4}.$

\item ${\left[ q^{2}+1{,4l+1,d\geq }\frac{{q}^{2}-4l+3}{2}\right] _{q^{2}}},$
where $q$ is an odd prime power such that $q\equiv 3$ $(mod$ $4)$ and $\frac{%
\left( q-1\right) \left( q-3\right) }{4}\leq l\leq \frac{{q}^{2}{-1}}{4}.$
\end{enumerate}

The rest of the paper is organized as follows. In Section 2, we present some
definitions and basic results of negacyclic codes. In Section 3, we
construct four families of LCD codes of length $n|\frac{{q-1}}{2},$ $n|\frac{%
{q+1}}{2}$ from negacyclic codes and we show that these codes are MDS.
Moreover, by studying their defining sets we determine parameters of a class
of LCD codes of length $n=q-1.$ In Section 4, we handle Hermitian negacyclic
LCD codes over $\mathbb{F}_{q^{2}}.$ The last Section is devoted to
conclusion.

\section{Preliminaries}

In this section, we will give some preliminaries, which are required for the
subsequent sections. Let $q$ be a prime power and $\mathbb{F}_{q}$ be the
finite field with $q$ elements. An $\left[ n,k\right] _{q}$ linear code $%
\mathcal{C}$ of length $n$ over $\mathbb{F}_{q}$ is a $k$-dimensional
subspace of the vector space $\mathbb{F}_{q}^{n}.$ The elements of $\mathcal{%
C}$ are of the form $\left( c_{0},c_{1},\ldots ,c_{n-1}\right) $ and called
codewords. The Hamming weight of any $c\in \mathcal{C}$ is the number of
nonzero coordinates of $c$ and denoted by $w\left( c\right) .$ The minimum
distance of $\mathcal{C}$ is defined as $d=\min \left\{ \left. w\left(
c\right) \right\vert 0\neq c\in \mathcal{C}\right\} .$ An $\left[ n,k\right]
_{q}$ linear code with minimum distance $d$ is said to be MDS (maximum
distance separable) if $n+1=k+d.$ The Euclidean dual code of $\mathcal{C}$
is defined to be
\begin{equation*}
\mathcal{C}^{\bot }=\left\{ \left. \mathbf{x}\in \mathbb{F}%
_{q}^{n}\right\vert \sum\limits_{i=0}^{n-1}x_{i}y_{i}=0,\forall \mathbf{y}%
\in \mathcal{C}\right\} .
\end{equation*}%
A code $\mathcal{C}$ is Euclidean self-orthogonal if $\mathcal{C}\subset
\mathcal{C}^{\bot }$ and Euclidean self-dual if $\mathcal{C}^{\bot }=%
\mathcal{C}.$

Let $\left\langle \mathbf{x,y}\right\rangle
=\sum\limits_{i=0}^{n-1}x_{i}y_{i}^{q}$ be the Hermitian inner product of $%
\mathbf{x}$ and $\mathbf{y\in }\mathbb{F}_{q^{2}}^{n}$ and $\mathcal{C}$ be
a code of length $n$ over $\mathbb{F}_{q^{2}}.$ The Hermitian dual code of $%
\mathcal{C}$ is defined to be
\begin{equation*}
\mathcal{C}^{\bot H}=\left\{ \left. \mathbf{x}\in \mathbb{F}%
_{q^{2}}^{n}\right\vert \sum\limits_{i=0}^{n-1}x_{i}y_{i}^{q}=0,\forall
\mathbf{y}\in \mathcal{C}\right\} .
\end{equation*}
A code $\mathcal{C}$ is Hermitian self-orthogonal if $\mathcal{C}\subset
\mathcal{C}^{\bot H}$ and Hermitian self-dual if $\mathcal{C}^{\bot H}=%
\mathcal{C}.$

A linear code $\mathcal{C}$ of length $n$ over $\mathbb{F}_{q}$ is said to
be negacyclic if for any codeword $(c_{0},c_{1},\ldots ,c_{n-1})\in \mathcal{%
C}$ we have that $(-c_{n-1},c_{0},c_{1},\ldots ,c_{n-2})\in \mathcal{C}.$ A
negacyclic code of length $n$ over $\mathbb{F}_{q}$ corresponds to a
principal ideal $\left\langle g\left( x\right) \right\rangle $ of the
quotient ring $\mathbb{F}_{q}\left[ x\right] /\left\langle {{x^{n}+1}}%
\right\rangle $ where $\left. g\left( x\right) \right\vert {{x^{n}+1}.}$ The
roots of the code $\mathcal{C}$ are the roots of the polynomial $g\left(
x\right) .$ So, if $\beta _{1},\beta _{2},\ldots ,\beta _{n-k}$ are the
zeros of $g\left( x\right) $ in the splitting field of ${{x^{n}}+1,}$ then $%
\mathbf{c}=\left( c_{0},c_{1},\ldots ,c_{n-1}\right) \in $$\mathcal{C}$ if
and only if $c\left( \beta _{1}\right) =c\left( \beta _{2}\right) =\ldots
=c\left( \beta _{n-k}\right) =0,$ where $c\left( x\right)
=c_{0}+c_{1}x+\ldots +c_{n-1}x^{n-1}.$

Let $2=ord_{q}\left( -1\right) $ and the multiplicative order of $q$ modulo $%
2n$ be $m.$ There exists $\delta \in \mathbb{F}_{q^{m}}^{\ast }$ called a
primitive $2n^{th}$ root of unity such that $\delta ^{n}=-1.$ Let $\zeta
=\delta ^{2},$ then $\zeta $ is a primitive $n^{th}$ root of unity.
Therefore, the roots of $x^{n}+1$ are $\left\{ \delta ,\delta ^{1+2},\ldots
,\delta ^{1+\left( n-1\right) 2}\right\} .$ Define $O_{2,n}\left( 1\right) $
as follows:%
\begin{equation}
O_{2,n}\left( 1\right) =\left\{ \left. 1+2i\right\vert 0\leq i\leq
n-1\right\} (mod\text{ }2n)\subseteq \mathbb{Z}_{2n}.
\end{equation}%
The defining set of the negacyclic code $\mathcal{C}$ is defined as
\begin{equation}
Z=\left\{ \left. 1+2i\in O_{2,n}\left( 1\right) \right\vert \delta ^{1+2i}%
\text{ is a root of }\mathcal{C}\right\} .
\end{equation}%
Clearly, $Z\subset O_{2,n}\left( 1\right) $ and the dimension of $\mathcal{C}
$ is $n-\left\vert Z\right\vert .$ Let $\mathbb{Z}_{2n}=\left\{ 0,1,2,\ldots
,2n-1\right\} $ denote the ring of integers modulo $2n.$ For any $s\in
\mathbb{Z}_{2n},$ the $q$-cyclotomic coset of $s$ modulo $2n$ is defined by $%
C_{s}=\left\{ \left. sq^{j}\left( {mod}\text{ }2n\right) \right\vert j\in
\mathbb{Z}\right\} .$

For each polynomial $g\left( x\right) =g_{0}+g_{1}x+\ldots +g_{r}x^{r}$ with
$g_{r}\neq 0,$ the reciprocal of $g\left( x\right) $ is defined to be the
polynomial $g^{\ast }\left( x\right) =x^{r}g\left( 1/x\right) .$ $g\left(
x\right) $ is called self-reciprocal if and only if $g\left( x\right)
=g^{\ast }\left( x\right).$

The following is a adapted version of BCH bound to negacyclic codes (\textit{for
general case see} \cite{Aydin,Krishna}).

\begin{theorem}
\label{thr0} \emph{(The BCH bound for negacyclic codes)} Let $\left( {n,q}%
\right) =1$ and also let $\delta $ be an $2n^{th}$ root of unity with $%
\delta ^{n}=-1.$ Then, the minimum distance of a negacyclic code of length $n
$ over $\mathbb{F}_{q}$ with the defining set $Z$ containing the set $%
\left\{ {\ \left. {1+2j}\right\vert l\leq j\leq l+d-2}\right\} $ is at least
$d$.
\end{theorem}

In the following, the relation between LCD codes and generator polynomial of
negacyclic codes is given.

\begin{theorem}
\cite{Zhu}\label{thr00} Let $\mathcal{C}=\left\langle g\left( x\right)
\right\rangle $ be a negacyclic code over $\mathbb{F}_{q}.$ Then, the
following statements are equivalent.

\begin{enumerate}
\item $\mathcal{C}$ is an LCD code.

\item $g\left( x\right) $ is self reciprocal.

\item $\delta ^{-1}$ is a root of $g\left( x\right) $ for every root $\delta
$ of $g\left( x\right) $ over the splitting field of $g\left( x\right) .$
\end{enumerate}
\end{theorem}

The following is a direct result of Theorem \ref{thr00}.

\begin{corollary}
\label{cor00} Euclidean LCD negacyclic codes over $\mathbb{F}_{q}$ of length
$n$\ exists if and only if $C_{s}=C_{-s}$ for some $s\in O_{2,n}\left(
1\right) ,$ where $C_{s}$ is a $q$-cyclotomic coset modulo $2n.$
\end{corollary}

In the following, we give some necessary information about Hermitian dual of
a negacyclic code over $\mathbb{F}_{q^{2}}.$ As a reference for example Refs. \cite{Yang1,Dinh1} can be used.

The following is an immediate result of Theorem 3.2 in \cite{Yang1}.

\begin{theorem}
\label{thr04} Let $\mathcal{C}=\left\langle g\left( x\right) \right\rangle $
be a negacyclic code over $\mathbb{F}_{q^{2}}.$ Then, the following
statements are equivalent.

\begin{enumerate}
\item $\mathcal{C}$ is an Hermitian LCD code.

\item $g\left( x\right) $ is Hermitian self reciprocal.

\item $\delta ^{-q}$ is a root of $g\left( x\right) $ for every root $\delta
$ of $g\left( x\right) $ over the splitting field of $g\left( x\right) .$
\end{enumerate}
\end{theorem}

The following corollary is a direct result of Theorem \ref{thr04}.

\begin{corollary}
\label{cor0} Hermitian LCD negacyclic codes over $\mathbb{F}_{q^{2}}$ of
length $n$\ exists if and only if $C_{s}=C_{-qs}$ for some $s\in
O_{2,n}\left( 1\right) ,$ where $C_{s}$ is a $q^{2}$-cyclotomic coset modulo
$2n.$
\end{corollary}

\section{New MDS LCD codes from negacyclic codes}

In this section, we aim to derive some classes of LCD codes from negacyclic
codes and to show that these codes are MDS. For this reason, we first
determine that the defining set $Z$ of negacyclic codes should satisfy $Z=-Z,$ and contain consecutive terms. Then, we construct MDS negacyclic LCD
codes from these negacyclic codes with defining set $Z$.

\subsection{New MDS negacyclic LCD codes of length $n$, where $n|\frac{{q-1}%
}{2}$}

Let $q$ be an odd prime power and let $n|\frac{{q-1}}{2},$ where $n\geq 3$.
It is clear that $q\equiv 1$ $(mod$ $2n)$ and then each $q$-cyclotomic coset
modulo $2n$ has exactly one element i.e., ${C_{1+2j}}=\left\{ {1+2j}\right\}
$ for all $0\leq j\leq n-1$. We also have the following result.

\begin{lemma}
\label{lm1} For all $0\leq j\leq n-1$, $-{C_{1+2j}}={C_{1+2\left( {n-j-1}%
\right) }}.$ Moreover, if $n$ is odd and $j=\frac{{n-1}}{2},$ then $-{%
C_{1+2j}}={C_{1+2j}}.$
\end{lemma}

\begin{proof}
$-\left( {1+2j}\right) \equiv 2n-\left( {1+2j}\right) =1+2\left( {n-j-1}%
\right) $ $(mod$ $2n).$ If $n$ is odd and $j=\frac{{n-1}}{2}$, then $j=n-j-1.
$
\end{proof}

In the following, we give the number of LCD negacyclic codes of length $n$
without proof.

\begin{corollary}
\label{cor1} If $n$ is even, then the number of nontrivial LCD negacyclic
codes of length $n$ is $2\left( {{2^{\frac{{n-2}}{2}}}-1}\right) .$ If $n$
is odd, then the number of nontrivial LCD negacyclic codes of length $n$ is $%
2\left( {{2^{\frac{{n-1}}{2}}}-1}\right) .$
\end{corollary}

For odd $n,$ we adjust the defining set
\begin{equation*}
{Z_{1}}=\bigcup\limits_{i=0}^{\lambda }\left( {{C_{1+2\left( {\frac{{n-1}}{2}%
+i}\right) }\cup C_{1+2\left( {\frac{{n-1}}{2}-i}\right) }}}\right) ,
\end{equation*}%
where $0\leq \lambda \leq \frac{{n-3}}{2}.$

For even $n,$ we establish the defining set
\begin{equation*}
{Z_{2}}=\bigcup\limits_{i=0}^{\lambda}\left( {{C_{1+2\left( {\frac{n}{2}+i}\right)
}}\cup {C_{1+2\left( {\frac{n}{2}-1-i}\right) }}}\right) ,
\end{equation*}%
where $1\leq \lambda\leq \frac{{n-2}}{2}.$ By Lemma \ref{lm1}, it is immediate
that $-{Z_{1}}=Z_{1}$ and $-{Z_{2}}=Z_{2}$ for each $0\leq \lambda \leq
\frac{{n-3}}{2}$ and $1\leq \lambda\leq \frac{{n-2}}{2},$ respectively. Now we are
ready to introduce two new classes of LCD negacyclic codes of length $n$
which are MDS.

In the following we give generalized version of Theorem 6.1 in \cite{Zhu}.

\begin{theorem}
\label{thr1} Let $q$ be an odd prime power and let $n|\frac{{q-1}}{2}$ such
that $n\neq 1.$ For even $n>2,$ a class of MDS negacyclic LCD codes with the
parameters ${\left[ {n,n-2\lambda-2 ,2\lambda +3}\right] _{q}},$ where $1\leq
\lambda \leq \frac{{n-2}}{2},$ exists. For odd $n,$ a class of MDS
negacyclic LCD codes with the parameters ${\left[ {n,n-2\lambda -1,2\left( {%
\lambda +1}\right) }\right] _{q}},$ where $0\leq \lambda \leq \frac{{n-3}}{2}%
,$ exists.
\end{theorem}

\begin{proof}
Let $n$ be even and, for each $1\leq {\lambda }\leq \frac{{n-2}}{2},$ define
$\mathcal{C}$ to be a negacyclic code of length $n$ having the defining set $%
Z_{2}$ over $\mathbb{F}_{q}$. Since, $Z_{2}$ consists of $2\lambda+2$
consecutive terms, where $1\leq \lambda \leq \frac{{n-2}}{2},$ the dimension
of $\mathcal{C}$ is $n-2\lambda-2 ,$ and by Theorem \ref{thr0} the minimum
distance of $\mathcal{C}$ is at least $2\lambda +3.$ Since $-Z_{2}=Z_{2}$
and $\mathcal{C}$ holds the definition of MDS codes, for each $1\leq {%
\lambda }\leq \frac{{n-2}}{2},$ $\mathcal{C}$ is MDS negacyclic LCD codes
with desired parameters.

Let $n$ be odd and, for each $0\leq \lambda \leq \frac{{n-3}}{2},$ define $%
\mathcal{C}$ to be a negacyclic code of length $n$ having the defining set $%
Z_{1}$ over $\mathbb{F}_{q}$. Since, $Z_{1}$ consists of $2\lambda +1$
consecutive terms, where $0\leq \lambda \leq \frac{{n-3}}{2},$ the dimension
of $\mathcal{C}$ is $n-2\lambda -1,$ and by Theorem \ref{thr0} the minimum
distance of $\mathcal{C}$ is at least $2\lambda +2$.$.$ Since $-Z_{1}=Z_{1}$
and $\mathcal{C}$ holds the definition of MDS codes, for each $0\leq \lambda
\leq \frac{{n-3}}{2},$ $\mathcal{C}$ is MDS negacyclic LCD codes with
desired parameters.
\end{proof}

\begin{example}
We present some parameters of MDS negacyclic LCD codes obtained by Theorem %
\ref{thr1} in Table \ref{table1}.
\end{example}
\begin{table}
\caption{Some MDS negacyclic LCD codes obtained from Theorem \protect\ref%
{thr1}}
\label{table1}\centering
\begin{tabular}{cccc}
\hline $q$ & $n$ & $\lambda $ & MDS Negacyclic LCD Codes \\
\hline $7$ & $3$ & $0$ & ${\left[ {3,2,2}\right] _{7}}$ \\
$9$ & $4$ & $1$ & ${\left[ {4,2,3}\right] _{9}}$ \\
$11$ & $5$ & $0\leq \lambda \leq 1$ & ${\left[ {5,4-2\lambda ,2\left( {%
\lambda +1}\right) }\right] _{11}}$ \\
$13$ & $6$ & $1\leq \lambda \leq 2$ & ${\left[ {6,6-2\lambda ,2\lambda +1}%
\right] _{13}}$ \\
$17$ & $8$ & $1\leq \lambda \leq 3$ & ${\left[ {8,8-2\lambda ,2\lambda +1}%
\right] _{17}}$ \\
$17$ & $4$ & $1$ & ${\left[ {4,2,3}\right] _{17}}$ \\
$19$ & $9$ & $0\leq \lambda \leq 3$ & ${\left[ {9,8-2\lambda ,2\left( {%
\lambda +1}\right) }\right] _{19}}$ \\
$19$ & $3$ & $0$ & ${\left[ {3,2,2}\right] _{19}}$ \\
\hline
\end{tabular}%
\end{table}

\subsection{New MDS negacyclic LCD codes of length $n$, where $n|\frac{{q + 1%
}}{2}$}

Let $q$ be an odd prime power and let $n|\frac{{q+1}}{2}$ such that $n\geq 3.
$ In this case, since $q\equiv -1$ $(mod$ $2n)$, each $q$-cyclotomic coset
modulo $2n$ has at most two elements. We give all $q$-cyclotomic cosets
modulo $2n$.

\begin{lemma}
\label{lm2} All $q$-cyclotomic cosets modulo $2n$ are given below.

\begin{enumerate}
\item If $n$ is even, then each cyclotomic coset has exactly two elements,
i.e., the cyclotomic coset ${C_{1 + 2j}}$ is the set $\left\{ {1 + 2j,\,1 +
2\left( {n - 1 - j} \right)} \right\}$ for all $0 \le j \le \frac{n}{2} - 1$.

\item If $n$ is odd, then each cyclotomic coset has exactly two elements
except for one, i.e., ${C_{1+2j}}=\left\{ {1+2j,\,1+2\left( {n-1-j}\right) }%
\right\} $ for all $0\leq j\leq \frac{n-3}{2},$ but ${C_{1+2j}}=\left\{ {1+2j%
}\right\} $ for $j=\frac{{n-1}}{2}.$
\end{enumerate}
\end{lemma}

\begin{proof}
Since $q\equiv -1$ $(mod$ $2n)$, we get $q\left( {1+2j}\right) \equiv
-1-2j=1+2\left( {n-1-j}\right) $ $(mod$ $2n)$. If $n$ is odd and $j=\frac{{%
n-1}}{2}$, then $j=n-1-j$ and ${C_{1+2j}}=\left\{ {1+2j}\right\} .$
\end{proof}

The following is immediately concluded from Lemma \ref{lm2}.

\begin{corollary}
\label{cor2} For all $q$-cyclotomic cosets modulo $2n$ containing $1+2j$, $-{%
C_{1+2j}}={C_{1+2j}}.$
\end{corollary}

We adjust the defining sets $Z_{1}$ and $Z_{2}$ with respect to the cases of
$n$ as follows: If $n$ is even, then ${Z_{1}}=\bigcup\limits_{j=\lambda }^{%
\frac{n}{2}-1}{{C_{1+2j},}}$ where $1\leq \lambda \leq \frac{n}{2}-1$. If $n$
is odd, then ${Z_{2}}=\bigcup\limits_{j=\lambda }^{\frac{{n-1}}{2}}{{%
C_{1+2j},}}$ where $1\leq \lambda \leq \frac{{n-1}}{2}.$
\newline
Then, from Lemma \ref{lm2}, we get
\begin{eqnarray*}
{Z_{1}} &=&\left\{ {1+2\lambda ,\,1+2\left( \lambda {+1}\right) ,\ldots
,1+2\left( {\frac{n}{2}-1}\right) ,\ldots ,1+2\left( {n-1-}\lambda \right) }%
\right\} ,\text{ }1\leq \lambda \leq \frac{n}{2}-1 \\
{Z_{2}} &=&\left\{ {1+2\lambda ,\,1+2\left( \lambda {+1}\right) ,\ldots
,1+2\left( {\frac{{n-1}}{2}}\right) ,\ldots ,1+2\left( {n-1-}\lambda \right)
}\right\} ,\text{ }1\leq \lambda\leq \frac{{n-1}}{2}.
\end{eqnarray*}%
So, $Z_{1}$ and $Z_{2}$ consists of exactly $n-2\lambda $ consecutive terms.
Moreover, Corollary \ref{cor2} implies that $-{Z_{1}}={Z_{1}}$ and $-{Z_{2}}=%
{Z_{2}}$. Now, we are ready to give new MDS negacyclic LCD codes of length $n
$ dividing $\frac{{q+1}}{2}$.

\begin{theorem}
\label{thr2} Let $q$ be an odd prime power and let $n|\frac{{q+1}}{2}$ such
that $n\neq 1$. For even $n>2,$ a family of MDS negacyclic LCD codes with
the parameters ${\left[ {n,2}\lambda {,n-2\lambda +1}\right] _{q}},$ where $%
1\leq \lambda \leq \frac{{n}}{2}-1,$ exists. For odd $n,$ a family of MDS
negacyclic LCD codes with the parameters ${\left[ {n,2\lambda ,n-2\lambda +1}%
\right] _{q}},$ where $1\leq \lambda \leq \frac{{n-1}}{2},$ exists.
\end{theorem}

\begin{proof}
Let $n$ be even and, for each $1\leq \lambda \leq \frac{{n}}{2}-1,$ define $%
\mathcal{C}$ to be a negacyclic code of length $n$ having the defining set $%
Z_{1}$ over $\mathbb{F}_{q}$. Since, $Z_{1}$ consists of $n-2\lambda $
consecutive terms, where $1\leq \lambda \leq \frac{{n-2}}{2},$the dimension
of $\mathcal{C}$ is $2\lambda ,$ and by Theorem \ref{thr0} the minimum
distance of $\mathcal{C}$ is at least $n-2\lambda +1.$ Since $-Z_{1}=Z_{1}$
and $\mathcal{C}$ holds the definition of MDS codes, for each $1\leq \lambda
\leq \frac{{n}}{2}-1,$ $\mathcal{C}$ is MDS negacyclic LCD codes with
desired parameters.

Let $n$ be odd and, for each $1\leq \lambda \leq \frac{{n-1}}{2},$ define $%
\mathcal{C}$ to be a negacyclic code of length $n$ having the defining set $%
Z_{2}$ over $\mathbb{F}_{q}$. Since, $Z_{2}$ consists of $n-2\lambda $
consecutive terms, where $1\leq \lambda \leq \frac{{n-1}}{2},$the dimension
of $\mathcal{C}$ is $2\lambda ,$ and by Theorem \ref{thr0} the minimum
distance of $\mathcal{C}$ is at least $n-2\lambda +1.$ Since $-Z_{2}=Z_{2}$
and $\mathcal{C}$ holds the definition of MDS codes, for each $1\leq \lambda
\leq \frac{{n-1}}{2},$ $\mathcal{C}$ is MDS negacyclic LCD codes with
desired parameters.
\end{proof}

\begin{example}
We present some parameters of MDS negacyclic LCD codes obtained by Theorem %
\ref{thr2} in Table \ref{table2}.
\end{example}

\begin{table}
\caption{Some MDS negacyclic LCD codes obtained from Theorem \protect\ref%
{thr2}}
\label{table2}\centering
\begin{tabular}{cccc}
\hline $q$ & $n$ & $\lambda $ & MDS Negacyclic LCD Codes \\
\hline $5$ & $3$ & $1$ & ${\left[ {3,2,2}\right] _{5}}$ \\
$7$ & $4$ & $1$ & ${\left[ {4,2,3}\right] _{7}}$ \\
$9$ & $5$ & $1\leq \lambda \leq 2$ & ${\left[ {5,2\lambda ,5-2\lambda +1}%
\right] _{9}}$ \\
$11$ & $6$ & $1\leq \lambda \leq 2$ & ${\left[ {6,2\lambda ,6-2\lambda +1}%
\right] _{11}}$ \\
$11$ & $3$ & $1$ & ${\left[ {3,2,2}\right] _{11}}$ \\
$13$ & $7$ & $1\leq \lambda \leq 3$ & ${\left[ {7,2\lambda ,7-2\lambda +1}%
\right] _{13}}$ \\
$17$ & $9$ & $1\leq \lambda \leq 4$ & ${\left[ {9,2\lambda ,9-2\lambda +1}%
\right] _{17}}$ \\
$17$ & $3$ & $1$ & ${\left[ {3,2,2}\right] _{17}}$ \\
$19$ & $10$ & $1\leq \lambda \leq 4$ & ${\left[ {10,2\lambda ,10-2\lambda +1}%
\right] _{19}}$ \\
$19$ & $5$ & $1\leq \lambda \leq 2$ & ${\left[ {5,2\lambda ,5-2\lambda +1}%
\right] _{19}}$ \\
\hline
\end{tabular}
\end{table}

\subsection{New negacyclic LCD codes of length $n=q+1$ with $\left.
4\right\vert n$}

Let $q$ be an odd prime power and let $n=q+1$ such that $4|n$. In this
subsection, we derive negacyclic codes of length $q+1$ which do not have to
be MDS. It follows from $q\not\equiv 1$ $(mod$ $2\left( {q+1}\right) )$ and $%
{q^{2}}\equiv 1$ $(mod$ $2\left( {q+1}\right) )$ that each $q$-cyclotomic
coset modulo $2n$ has at most two elements. Suppose $q\left( {1+2j}\right)
\equiv \left( {1+2j}\right) $ $(mod$ $2\left( {q+1}\right) )$ for some $%
0\leq j\leq n-1.$ Since $\left( {q-1,q+1}\right) =2$ and $1\leq 1+2j\leq
2n-1,$ we get $2j=q,$ which is a contradiction. This implies that each $q$%
-cyclotomic coset modulo $2\left( {q+1}\right) $ has precisely two elements.
We give an exact characterization for all $q$-cyclotomic cosets modulo $2n$.

\begin{lemma}
\label{lm3} All $q$-cyclotomic cosets modulo $2\left({q+1}\right)$ are as
follows.

\begin{enumerate}
\item ${C_{1+2j}}=\left\{ {1+2j,\,1+2\left( {\frac{{q-1}}{2}-j}\right) }%
\right\} ,$ for all $0\leq j\leq \frac{{q-3}}{4}.$

\item ${C_{1+2j}}=\left\{ {1+2j,\,1+2\left( {n+\frac{{q-1}}{2}-j}\right) }%
\right\} ,$ for all $\frac{{q+1}}{2}\leq j\leq \frac{{3q-1}}{4}.$
\end{enumerate}
\end{lemma}

\begin{proof}
\begin{enumerate}
\item If $0\leq j\leq \frac{{q-3}}{4},$ then $j<\frac{{q-1}}{2}.$ Since $2qj\equiv
-2j$ $(mod$ $2n),$ $q\left( {1+2j}\right) \equiv q-2j=1+2\left( {\frac{{q-1}%
}{2}-j}\right) $ $(mod$ $2n).$
\item If $\frac{{q+1}}{2}\leq j\leq \frac{{3q-1}}{4}%
,$ then $j>\frac{{q-1}}{2}$ and so $q\left( {1+2j}\right) \equiv 1+2\left( {%
n+\frac{{q-1}}{2}-j}\right) $ $(mod$ $2n).$
\end{enumerate}
The union of all $q$-cyclotomic cosets given here makes up the set ${O_{2,q+1}}\left( 1\right) $ and so the
proof is completed.
\end{proof}

\begin{lemma}
\label{lm4} For all $0 \le j \le \frac{{q - 3}}{4}$, $- {C_{1 + 2j}} = {C_{1
+ 2\left( {\frac{{q + 1}}{2} + j} \right)}}$.
\end{lemma}

\begin{proof}
It can be seen that $-\left( {1+2j}\right) \equiv 1+2\left( {n-j-1}\right) $ $(mod$ $%
2n).$ From Lemma \ref{lm3}$(1),$ it is enough to find a integer $k$ such that $n+\frac{{q+1}%
}{2}-k=n-j-1$. This is possible only when $k=\frac{{q+1}}{2}+j.$
\end{proof}

As a result of Lemma \ref{lm4}, one can see that ${C_{1+2j}}\neq -{C_{1+2j}}$
for all $0\leq j\leq n-1$. We establish the defining set $Z$ to be $Z=\left(
\bigcup\limits_{j=\lambda }^{\frac{{q-3}}{4}}{{C_{1+2j}}}\right) \cup \left(
\bigcup\limits_{j=\frac{{q+1}}{2}+\lambda }^{\frac{{3q-1}}{4}}{{C_{1+2j}}}%
\right) ,$ where $1\leq \lambda \leq \frac{{q-3}}{4}.$ Then, by Lemma \ref%
{lm3}, we have
\begin{equation*}
Z=\left\{
\begin{array}{c}
{1+2\lambda ,1+2\left( {\lambda +1}\right) ,\ldots ,1+2\left( {\frac{{q-1}}{2%
}-\lambda }\right) ,} \\
{1+2\left( {\frac{{q+1}}{2}+\lambda }\right) ,1+2\left( {\frac{{q+3}}{2}%
+\lambda }\right) ,\ldots ,1+2\left( {q-\lambda }\right) }%
\end{array}%
\right\} {.}
\end{equation*}

Clearly, $Z$ contains $\frac{{q+1}}{2}-2{\lambda }$ consecutive terms and $%
\left\vert Z\right\vert =q+1-4{\lambda }.$ These facts provide us to derive
a class of LCD negacyclic codes.

\begin{theorem}
\label{thr5} Assume that $q$ is an odd prime power and $n=q+1$ such that $%
4|n $. Then, a class of LCD negacyclic codes with parameters
\begin{equation*}
{\left[ {q + 1,4\lambda,d \ge \frac{{q + 3}}{2} - 2\lambda} \right]_q}
\end{equation*}

where $1\leq \lambda \leq \frac{{q-3}}{4},$ exists.
\end{theorem}

\begin{proof}
Let $\mathcal{C}$ be a negacyclic code of length $q+1$ with defining set $Z$
over $\mathbb{F}_{q}.$ The parameters of $\mathcal{C}$ are followed from
that $Z$ contains $\frac{{q+1}}{2}-2\lambda $ consecutive terms and $%
\left\vert Z\right\vert =q+1-4\lambda $.
\end{proof}

\begin{example}
We list some parameters of negacyclic LCD codes acquired by Theorem \ref%
{thr5} in Table \ref{table3}.
\end{example}

\begin{table}
\caption{Some negacyclic LCD codes obtained from Theorem \protect\ref{thr5}}
\label{table3}\centering
\begin{tabular}{cccc}
\hline $q$ & $n$ & $\lambda $ & Negacyclic LCD Codes \\
\hline $19$ & $20$ & $1\leq \lambda \leq 4$ & ${\left[ {20,4\lambda ,\geq
11-2\lambda }\right] _{19}}$ \\
$23$ & $24$ & $1\leq \lambda \leq 5$ & ${\left[ {24,4\lambda ,\geq
13-2\lambda }\right] _{23}}$ \\
\hline
\end{tabular}%
\end{table}

\section{Negacyclic MDS Hermitian LCD Codes}

In this section, we study Hermitian LCD codes over finite fields $\mathbb{F}_{q^{2}}.$ We use negacyclic codes of length $n,$ where $n|q-1$ and $%
n=q^{2}+1$ to construct $q^{2}$-ary MDS Hermitian LCD codes and Hermitian
LCD codes. To accomplish this task, we need to determine the defining set $%
\overline{Z}$ of negacyclic codes and the number of consecutive terms
contained by $\overline{Z},$ where $\overline{Z}=-q\overline{Z}.$ At the
beginning, we determine exact structure of $q^{2}$-cyclotomic cosets modulo $%
2n.$

\subsection{Negacyclic MDS Hermitian LCD codes of length $\left.
n\right\vert \left( q-1\right) $}

In this subsection, we use negacyclic codes of length $n=\frac{q-1}{\gamma }$ to construct MDS Hermitian LCD codes, where $q$ is an
odd prime power. Since $\left. n\right\vert \left( q^{2}-1\right) $ and $%
\gcd \left( n,q+1\right) =1$ or $2$ we have that $q^{2}=1+\gamma n\left(
q+1\right) \equiv 1\left( mod\text{ }2n\right) $ or $q^{2}=1+2\gamma n\frac{%
\left( q+1\right) }{2}\equiv 1$ $\left( mod\text{ }2n\right) .$ This means
that each $q^{2}$-cyclotomic coset modulo $2n$ has only one element.

The following enables us to determine the number of elements and consecutive
terms contained by the defining set $\overline{Z}$ which we define later.

\begin{lemma}
\label{lm6}Let $n=\frac{q-1}{\gamma }.$

\begin{enumerate}
\item If $2\nmid \gamma ,$ then for all $j\leq \frac{n}{2}-1$ we have that $%
-qC_{1+2j}=C_{1+2\left( \frac{n}{2}-1-j\right) }$ and for all $j>\frac{n}{2}$
we have that $-qC_{1+2j}=C_{1+2\left( \frac{3n}{2}-1-j\right) }.$

\item If $2\left\vert \gamma \right. ,$ then for all $j\leq n-1$ we have
that $-qC_{1+2j}=C_{1+2\left( n-1-j\right) }.$
\end{enumerate}
\end{lemma}

\begin{proof}
\begin{enumerate}
\item Observe that $-q\equiv n-1$ $\left( mod\text{ }2n\right) .$ Then we
have $-q\left( 1+2j\right) =-q-2qj\equiv n-1+2\left( n-1\right) j\equiv
1+n-2-2j=1+2\left( \frac{n}{2}-1-j\right) $ $\left( mod\text{ }2n\right) .$
If $j>\frac{n}{2}$, then $\frac{n}{2}-1-j<0$ and so we have that $1+2\left(
\frac{n}{2}-1-j\right) \equiv 1+2\left( \frac{n}{2}-1-j\right) +2n=1+2\left(
\frac{3n}{2}-1-j\right) $ $\left( mod\text{ }2n\right) .$

\item In this case $-q\equiv -1$ $\left( mod\text{ }2n\right) .$ So, we have
$-q\left( 1+2j\right) \equiv -1-2j\equiv 1-2-2j+2n=1+2\left( n-1-j\right) $ $%
\left( mod\text{ }2n\right) .$
\end{enumerate}
\end{proof}

Let $n=\frac{q-1}{\gamma }>4$ and $q\equiv 1$ $\left( mod\text{ }4\right) .$
Then, we give the defining set $\overline{Z}$ as below:

\begin{equation*}
\begin{array}{c}
\text{If }2\nmid \gamma ,\text{ then }\overline{Z}=\bigcup\limits_{j=\frac{%
q-4\gamma -1}{4\gamma }-l}^{\frac{{q-1}}{4\gamma }+l}{{C_{1+2j}}},\text{
where }0\leq l\leq \frac{q-4\gamma -1}{4\gamma }. \\
\text{If }2\left\vert \gamma \right. \text{ and }n\text{ is even, then }%
\overline{Z}=\bigcup\limits_{j=\frac{q-2\gamma -1}{2\gamma }-l}^{\frac{{q-1}%
}{2\gamma }+l}{{C_{1+2j}}},\text{ where }0\leq l\leq \frac{q-4\gamma -1}{%
2\gamma }. \\
\text{If }2\left\vert \gamma \right. \text{ and }n\text{ is odd, then }%
\overline{Z}=\bigcup\limits_{j=\frac{q-\gamma -1}{2\gamma }-l}^{\frac{{%
q-\gamma -1}}{2\gamma }+l}{{C_{1+2j}}},\text{ where }0\leq l\leq \frac{%
q-3\gamma -1}{2\gamma }.%
\end{array}%
\end{equation*}

From the definitions of the defining sets $\overline{Z},$ we give the number
of elements of the defining sets $\overline{Z}$ and we show that all of the
elements are consecutive.

\begin{equation*}
\begin{array}{c}
\text{If }2\nmid \gamma ,\text{ then }\left\vert \overline{Z}\right\vert
=2l+2,\text{ where }0\leq l\leq \frac{q-4\gamma -1}{4\gamma }. \\
\text{If }2\left\vert \gamma \right. \text{ and }n\text{ is even, then }%
\left\vert \overline{Z}\right\vert =2l+2,\text{ where }0\leq l\leq \frac{%
q-4\gamma -1}{2\gamma}. \\
\text{If }2\left\vert \gamma \right. \text{ and }n\text{ is odd, then }%
\left\vert \overline{Z}\right\vert =2l+1,\text{ where }0\leq l\leq \frac{%
q-3\gamma -1}{2\gamma }.%
\end{array}%
\end{equation*}%
So, the following is immediate.

\begin{theorem}
\label{thr6}Let $q\equiv 1$ $\left( mod\text{ }4\right) ,$ and $n=\frac{{q-1}%
}{\gamma }>2.$

\begin{enumerate}
\item If $2\nmid \gamma ,$ then there exists a $q^{2}-$ary ${\left[ n{%
,n-2l-2,2l+3}\right] }$ negacyclic MDS Hermitian LCD code, where $0\leq
l\leq \frac{q-4\gamma -1}{4\gamma }.$

\item If $2\left\vert \gamma \right. $ and $n$ is even, then there exists a $%
q^{2}-$ary ${\left[ n{,n-2l-2,2l+3}\right] }$ negacyclic MDS Hermitian LCD
code, where $0\leq l\leq \frac{q-4\gamma -1}{2\gamma }.$

\item If $2\left\vert \gamma \right. $ and $n$ is odd, then there exists a $%
q^{2}-$ary ${\left[ n{,n-2l-1,2l+2}\right] }$ negacyclic MDS Hermitian LCD
code, $0\leq l\leq \frac{q-3\gamma -1}{2\gamma }.$
\end{enumerate}
\end{theorem}

\begin{example}
Let $q=29$ and $n=\frac{{28}}{\lambda }.$ Then, by applying Theorem \ref%
{thr6}, we obtain $17$ $q^{2}$-ary negacyclic MDS Hermitian LCD codes with
parameters $\left[ 28,26,3\right] ,$ $\left[ 28,24,5\right] ,$ $\left[
28,22,7\right] ,$ $\left[ 28,20,9\right] ,$ $\left[ 28,18,11\right] ,$ $%
\left[ 28,16,13\right] ,$ $\left[ 28,14,15\right] ,$ $\left[ 14,12,3\right]
, $ $\left[ 14,10,5\right] ,$ $\left[ 14,8,7\right] ,$ $\left[ 14,6,9\right]
, $ $\left[ 14,4,11\right] ,$ $\left[ 14,2,13\right] ,$ $\left[ 7,6,2\right]
,$ $\left[ 7,4,4\right] ,$ $\left[ 7,2,6\right] ,$ $\left[ 4,3,2\right] .$
\end{example}

\begin{table}
\caption{Some Hermitian MDS negacyclic LCD codes obtained from Theorem
\protect\ref{thr6}}
\label{table4}\centering
\begin{tabular}{ccccc}
\hline $q$ & $n$ & $\gamma $ & $l$ & Negacyclic MDS Hermitian LCD Codes \\
\hline $5$ & $4$ & $1$ & $l=0$ & ${\left[ {5,4,}2\right] _{5^{2}}}$ \\
$13$ & $12$ & $1$ & $0\leq l\leq 2$ & ${\left[ {12,12-2l-2,}2l+3\right]
_{13^{2}}}$ \\
$13$ & $6$ & $2$ & $0\leq l\leq 1$ & ${\left[ {6,6-2l-2,}2l+3\right]
_{13^{2}}}$ \\
$17$ & $16$ & $1$ & $0\leq l\leq 3$ & ${\left[ {16,16-2l-2,}2l+3\right]
_{17^{2}}}$ \\
$17$ & $8$ & $2$ & $0\leq l\leq 2$ & ${\left[ {8,8-2l-2,}2l+3\right]
_{17^{2}}}$ \\
$17$ & $4$ & $4$ & $l=0$ & ${\left[ {4,2,}3\right] _{17^{2}}}$ \\
\hline
\end{tabular}%
\end{table}

By expanding the defining set $\overline{Z},$ we can obtain non MDS
negacyclic Hermitian LCD codes over $\mathbb{F}_{q^{2}}.$ Let $n=\frac{{q-1}%
}{\gamma }>4,$ $2\nmid \gamma $ and $q\equiv 1$ $\left( mod\text{ }4\right)
. $ Then adjust the defining set $\overline{Z}$ as below.%
\begin{equation*}
\overline{Z}=\left( \bigcup\limits_{j=0}^{\frac{{q-2\gamma -1}}{2\gamma }}{{%
C_{1+2j}}}\right) \cup \left( \bigcup\limits_{j=\frac{q-1}{2\gamma }+l}^{%
\frac{{q-\gamma -1}}{\gamma }-l}{{C_{1+2j}}}\right) ,\text{ where }0\leq
l\leq \frac{q-8\gamma -1}{4\gamma }.
\end{equation*}

\begin{theorem}
\label{thr7}Let $n=\frac{{q-1}}{\gamma }>4,$ $2\nmid \gamma $ and $q\equiv 1$
$\left( mod\text{ }4\right) .$ Then for each $0\leq l\leq \frac{q-8\gamma -1%
}{4\gamma }$ there exists a $q^{2}-$ary ${\left[ n{,n-}\left( {\frac{{q-1}}{2%
}+2l+2}\right) {,d\geq \frac{{q-1}}{2{\gamma }}+l+2}\right] }$ negacyclic
Hermitian LCD code.
\end{theorem}

Let $n=\frac{q-1}{\gamma }>2$ and $q\equiv 3$ $\left( mod\text{ }4\right) .$
Then, we can give the defining set $\overline{Z}$ as the following.

\begin{equation*}
\begin{array}{c}
\text{If }2\nmid \gamma ,\text{ then }\overline{Z}=\bigcup\limits_{j=\frac{%
q-2\gamma -1}{4\gamma }-l}^{\frac{q-2\gamma -1}{4\gamma }+l}{{C_{1+2j}}},%
\text{ where }0\leq l\leq \frac{q-2\gamma -1}{4\gamma }. \\
\text{If }2\left\vert \gamma \right. ,\text{ then }\overline{Z}%
=\bigcup\limits_{j=\frac{q-\gamma -1}{2\gamma }-l}^{\frac{q-\gamma -1}{%
2\gamma }+l}{{C_{1+2j}}},\text{ where }0\leq l\leq \frac{q-3\gamma -1}{%
2\gamma }.%
\end{array}%
\end{equation*}

By the definition of $\overline{Z},$ the cardinality of $\overline{Z},$ $%
\left\vert \overline{Z}\right\vert$ is as the following and all its elements
are consecutive.%
\begin{equation*}
\begin{array}{c}
\text{If }2\nmid \gamma ,\text{ then }\left\vert \overline{Z}\right\vert
=2l+1,\text{ where }0\leq l\leq \frac{q-2\gamma -1}{4\gamma }. \\
\text{If }2|\gamma ,\text{ then }\left\vert \overline{Z}\right\vert =2l+1,%
\text{ where }0\leq l\leq \frac{q-3\gamma -1}{2\gamma }.%
\end{array}%
\end{equation*}%
Thus, the following is immediate.

\begin{theorem}
\label{thr6a}Let $q\equiv 3$ $\left( mod\text{ }4\right),$ and $n=\frac{{q-1}%
}{\gamma }>2.$

\begin{enumerate}
\item If $2\nmid \gamma ,$ then there exists a $q^{2}-$ary ${\left[ n
,n-2l-1,2l+2\right] }$ negacyclic MDS Hermitian LCD code, where $0\leq l\leq
\frac{q-2\gamma -1}{4\gamma }.$

\item If $2\left\vert \gamma \right. ,$ then there exists a $q^{2}-$ary ${%
\left[ n,n-2l-1,2l+2\right] }$ negacyclic MDS Hermitian LCD code, where $%
0\leq l\leq \frac{q-3\gamma -1}{2\gamma }.$
\end{enumerate}
\end{theorem}

\begin{table}
\caption{Some Hermitian MDS negacyclic LCD codes obtained from Theorem
\protect\ref{thr6a}}
\label{table5}\centering
\begin{tabular}{ccccc}
\hline $q$ & $n$ & $\gamma $ & $l$ & Negacyclic MDS Hermitian LCD Codes \\
\hline $7$ & $6$ & $1$ & $0\leq l\leq 1$ & ${\left[ {6,6-2l-1,}2l+2\right]
_{7^{2}}} $ \\
$11$ & $10$ & $1$ & $0\leq l\leq 2$ & ${\left[ {10,10-2l-1,}2l+2\right]
_{11^{2}}}$ \\
$11$ & $5$ & $2$ & $0\leq l\leq 1$ & ${\left[ {5,5-2l-1,}2l+2\right]
_{11^{2}}}$ \\
$19$ & $18$ & $1$ & $0\leq l\leq 4$ & ${\left[ {18,18-2l-1,}2l+2\right]
_{19^{2}}}$ \\
$19$ & $9$ & $2$ & $0\leq l\leq 3$ & ${\left[ {9,9-2l-1,}2l+2\right]
_{19^{2}}}$ \\
$19$ & $6$ & $3$ & $0\leq l\leq 1$ & ${\left[ {6,6-2l-1,}2l+2\right]
_{19^{2}}}$ \\
\hline
\end{tabular}%
\end{table}

Let $n=\frac{{q-1}}{\gamma }>4,$ $2\nmid \gamma $ and $q\equiv 3$ $\left( mod%
\text{ }4\right) .$ Then we establish the defining set $\overline{Z}$ as
follows.%
\begin{equation*}
\overline{Z}=\left( \bigcup\limits_{j=0}^{\frac{{q-2\gamma -1}}{2}}{{C_{1+2j}%
}}\right) \cup \left( \bigcup\limits_{j=\frac{q-1}{2\gamma }+l}^{\frac{{%
q-\gamma -1}}{\gamma }-l}{{C_{1+2j}}}\right) ,\text{ where }0\leq l\leq
\frac{q-6\gamma -1}{4\gamma }.
\end{equation*}

Now, we can construct non MDS negacyclic Hermitian LCD codes over $\mathbb{F}%
_{q^{2}}.$

\begin{theorem}
\label{thr7a}Let $n=\frac{{q-1}}{\gamma }>4,$ $2\nmid \gamma $ and $q\equiv 3
$ $\left( mod\text{ }4\right) .$ Then for each $0\leq l\leq \frac{q-6\gamma
-1}{4\gamma }$ there exists a $q^{2}-$ary ${\left[ n{,n-}\left( {\frac{{q-1}%
}{2}+2l+2}\right) {,d\geq \frac{{q-1}}{2}+l+2}\right] }$ negacyclic
Hermitian LCD code.
\end{theorem}

\subsection{Negacyclic Hermitian LCD codes of length $n=q^{2}+1$}

In this subsection, we use negacyclic codes of length $n=q^{2}+1$ to
construct Hermitian LCD codes, where $q$ is an odd prime power. The
following is similar to Lemma 4.1 in \cite{Kai}.

\begin{lemma}
\label{lm8}Let $n=q^{2}+1.$ Then, the $q^{2}$-cyclotomic cosets modulo $2n$
containing odd integers from $1$ to $2n$ are $C_{1+2j}=\left\{ {1+2j,n-1-2j}%
\right\} ,$ $0\leq j<\frac{n-2}{4},C_{1+2j}=\left\{ {1+2j}\right\} $, $j=%
\frac{n-2}{4}$, $C_{1+2j}=\left\{ {1+2j,3n-1-2j}\right\} ,$ $\frac{n}{2}\leq
j<\frac{3n-2}{4}$, and $C_{1+2j}=\left\{ {1+2j}\right\} ,$ $j=\frac{3n-2}{4}%
. $
\end{lemma}

For the length $n=q^{2}+1$ we consider two cases. The first case is $q\equiv
1$ $(mod$ $4).$ We establish the defining set $\overline{Z}$ as $\overline{Z}%
=\overline{Z}_{1}\cup -q\overline{Z}_{1}$ where $\overline{Z}%
_{1}=\bigcup\limits_{j=l}^{\frac{{q}^{2}{-1}}{4}}{{C_{1+2j}}},$ $\frac{%
\left( q-1\right) ^{2}}{4}\leq l\leq \frac{{q}^{2}{-1}}{4}.$ In \cite{Kai}
it was shown that $\overline{Z}_{1}\cap -q\overline{Z}_{1}=\emptyset .$
Therefore, the cardinality of the defining set $\overline{Z}$ is $\left\vert
\overline{Z}\right\vert =4\left( \frac{{q}^{2}{-1}}{4}-l\right) +2={q}%
^{2}-4l+1.$ Furthermore, $\overline{Z}$ contains at least $\frac{{q}^{2}-4l+1%
}{2}$ consecutive terms. Then, we have the following result.

\begin{theorem}
\label{thr8}Let $q$ be an odd prime power with $q\equiv 1$ $(mod$ $4).$
Then, there exists a class of $q^{2}$-ary negacyclic Hermitian LCD codes
with parameters ${\left[ q^{2}+1{,4l,d\geq }\frac{{q}^{2}-4l+3}{2}\right] },$
where $\frac{\left( q-1\right) ^{2}}{4}\leq l\leq \frac{{q}^{2}{-1}}{4}.$
\end{theorem}

The other case is $q\equiv 3$ $(mod$ $4).$

\begin{lemma}
\label{lm8a}Let $n=q^{2}+1$ and $q\equiv 3$ $(mod$ $4).$ Then, we have the
following.

\begin{enumerate}
\item For $j=\frac{{q}^{2}{-1}}{4},$ $C_{1+2j}=-qC_{1+2j}.$

\item For all $\frac{\left( q-1\right) \left( q-3\right) }{4}\leq j,k<\frac{{%
q}^{2}{-1}}{4},$ $C_{1+2k}\neq -qC_{1+2j}.$
\end{enumerate}
\end{lemma}

\begin{proof}
\begin{enumerate}
\item For $j=\frac{{q}^{2}{-1}}{4},$ $C_{1+2j}=\left\{ \frac{{q}^{2}+{1}}{4}%
\right\} .$ Since $\left. 4\right\vert \left( q+1\right) ,$ $\left(
q+1\right) \frac{\left( q^{2}+1\right) }{2}\equiv 0$ $(mod$ $2n)$ and so $-q%
\frac{\left( q^{2}+1\right) }{2}\equiv \frac{\left( q^{2}+1\right) }{2}$ $%
(mod$ $2n).$

\item Assume otherwise. Then, for some $\frac{\left( q-1\right) \left(
q-3\right) }{4}\leq j,k\leq \frac{\left( q-1\right) \left( q+5\right) }{4}$
except for $\frac{{q}^{2}{-1}}{4},$ $-q\left( 1+2j\right) \equiv 1+2k$ $(mod$
$2n)$ or $\frac{{q}+{1}}{2}+k+qj\equiv 0$ $(mod$ $n).$ It follows from $%
\frac{\left( q-1\right) \left( q-3\right) }{4}\leq j,k\leq \frac{\left(
q-1\right) \left( q+5\right) }{4}$ that $\frac{\left( q-3\right) }{4}n+2\leq
\frac{{q}+{1}}{2}+k+qj\leq \frac{\left( q+5\right) }{4}n-2.$ This implies
that the possible value of $\frac{{q}+{1}}{2}+k+qj$ is only $\frac{\left(
q+1\right) }{4}n,$ which is possible only when $k=j=\frac{{q}^{2}{-1}}{4}.$
This contradicts with the choice of $j$ and $k.$
\end{enumerate}
\end{proof}

For the case $q\equiv 3$ $(mod$ $4),$ we adjust the defining set $\overline{Z%
}$ as $\overline{Z}=\overline{Z}_{2}\cup -q\overline{Z}_{2},$ where $%
\overline{Z}_{2}=\bigcup\limits_{j=l}^{\frac{{q}^{2}{-1}}{4}}{{C_{1+2j}}},$ $%
\frac{\left( q-1\right) \left( q-3\right) }{4}\leq l\leq \frac{{q}^{2}{-1}}{4%
}.$ By Lemma \ref{lm8a}, the cardinality of the defining set $\overline{Z}$
is $\left\vert \overline{Z}\right\vert =4\left( \frac{{q}^{2}{-1}}{4}%
-l\right) +1={q}^{2}-4l.$ Additionally, $\overline{Z}$ contains at least $%
\frac{{q}^{2}-4l+1}{2}$ consecutive terms. Then, we have the following
result.

\begin{theorem}
\label{thr8a}Let $q$ be an odd prime power with $q\equiv 3$ $(mod$ $4).$
Then, there exists a class of $q^{2}-$ary negacyclic Hermitian LCD codes
with parameters ${\left[ q^{2}+1{,4l+1,d\geq }\frac{{q}^{2}-4l+3}{2}\right] }%
,$ where $\frac{\left( q-1\right) \left( q-3\right) }{4}\leq l\leq \frac{{q}%
^{2}{-1}}{4}.$
\end{theorem}

\begin{table}
\caption{Some negacyclic Hermitian LCD codes obtained from Theorems \protect
\ref{thr8} and \protect\ref{thr8a}}
\label{table6}\centering
\begin{tabular}{cccc}
\hline $q$ & $n$ & $l$ & Negacyclic Hermitian LCD Codes \\
\hline $3$ & $10$ & $0\leq l\leq 2$ & ${\left[ {10,4l+1,\geq }\frac{{9}%
-4l+3}{2}\right] _{3^{2}}}$ \\
$5$ & $26$ & $4\leq l\leq 6$ & ${\left[ {26,4l,d\geq }\frac{{25}-4l+3}{2}%
\right] _{5^{2}}}$ \\
$7$ & $50$ & $6\leq l\leq 12$ & ${\left[ {50,4l+1,d\geq }\frac{{49}-4l+3}{2}%
\right] _{7^{2}}}$ \\
$13$ & $170$ & $36\leq l\leq 42$ & ${\left[ {170,4l,d\geq }\frac{1{69}-4l+3}{%
2}\right] _{13^{2}}}$ \\
\hline
\end{tabular}%
\end{table}

\section{Conclusion}

In this paper, we study some classes of MDS negacyclic LCD codes of length $%
n|\frac{{q-1}}{2},$ $n|\frac{{q+1}}{2}$ and some classes of negacyclic LCD
codes of length $n=q+1$. In Theorem \ref{thr1} we give a corrected and
generalized of the result of Theorem 6.1 in \cite{Zhu}. We also obtain some
classes of $q^{2}$-ary Hermitian MDS negacyclic LCD codes of length $%
n|\left( q-1\right) $ and some classes of $q^{2}$-ary Hermitian negacyclic
LCD codes $n=q^{2}+1.$ We remark that the parameters of Hermitian LCD codes, which was given in
\cite{Galvez,Li2}, haven't covered ones given in this paper.


\begin{thebibliography}{33}
\bibitem{Aydin}  {\em N. Aydin, I. Siap and D. K. Ray-Chaudhuri}, The structure of 1-generator quasi-twisted codes and new linear codes, Design Code Cryptogr {\bf 24} (2001) 313-326.
\bibitem{Beelen}  {\em P. Beelen and L. Jin}, Explicit MDS codes with complementary duals, IEEE Trans. Inform. Theory, (2018) https://doi.org/10.1109/TIT.2018.2816934.
\bibitem{Carlet}  {\em C. Carlet, and S. Guilley}, Complementary dual codes for counter-measures to side-channel attacks, In: E. R. Pinto et al. (eds), In Coding Theory and Applications, CIM Series in Mathematical Sciences, {\bf 3} (2014) 97-105.
\bibitem{Carlet1}  {\em C. Carlet, S. Mesnager, C. Tang and Y. Qi}, Euclidean and Hermitian LCD MDS codes, Design Code Cryptogr, (2018) https://doi.org/10.1007/s10623-018-0463-8.
\bibitem{Chen}  {\em B. Chen, H. Q. Dinh, and H. Liu}, Repeated-root constacyclic codes of length $2\ell ^{m}p^{n}.$, Finite Fields Th App, {\bf 33} (2015) 137-159.
\bibitem{Chen1}  {\em B. Chen and H. Liu}, New constructions of MDS codes with complementary duals, IEEE Trans. Inform. Theory, {\bf 64} (2017) 5776 - 5782.
\bibitem{Dinh1}  {\em H. Q. Dinh}, Constacyclic codes of length $p^{s}$ over $\mathbb{F}_{p^{m}}+u\mathbb{F}_{p^{m}}$, J. Algebra, {\bf 324} (2010) 940-950.
\bibitem{Dougherty}  {\em S. T. Dougherty, J. L. Kim, B. Ozkaya, L. Sok, and P. Sol\'{e}}, The combinatorics of LCD codes: Linear Programming bound and orthogonal matrices, Int. J. Inf. Coding Theory, {\bf 4} (2017) 116-128.
\bibitem{Esmaeili}  {\em M. Esmaeili, and S. Yari}, On complementary-dual quasi-cyclic codes, Finite Fields Th App, {\bf 15} (2009) 375-386.
\bibitem{Galvez}  {\em L. Galvez, J. L. Kim, N. Lee, Y. G. Roe and B. S. Won}, Some bounds on binary LCD codes, Cryptogr. Commun., {\bf 10} (2017) 719-728.
\bibitem{Guneri}  {\em C. Guneri, B. Ozkaya and P. Sol\'{e}}, Quasi-Cyclic Complementary Dual Code, Finite Fields Th App, {\bf 42} (2016) 67-80.
\bibitem{Jin}  {\em L. Jin}, Construction of MDS codes with complementary duals, IEEE Trans. Inform. Theory, {\bf 63} (2017) 2843-2847.
\bibitem{Kai}  {\em X. Kai and S. Zhu}, New quantum MDS codes from negacyclic codes, IEEE Trans. Inform. Theory, {\bf 59} (2013) 1193-1197.
\bibitem{Krishna}  {\em A. Krishna and D. V. Sarwate}, Pseudocyclic maximum-distance-separable codes, IEEE Trans. Inform. Theory {\bf 36} (1990) 880-884.
\bibitem{Li2}  {\em C. Li}, Hermitian LCD codes from cyclic codes, Design Code Cryptogr, {\bf 86} (2018) 2261-2278.
\bibitem{Li1}  {\em C. Li, C. Ding and S. Li}, LCD cyclic codes over finite fields, IEEE Trans. Inform. Theory, {\bf 63} (2017) 4344-4356.
\bibitem{Li}  {\em S. Li, C. Li, C. Ding and H. Liu}, Two families of LCD BCH codes, IEEE Trans. Inform. Theory, {\bf 63} (2017) 5699-5717.
\bibitem{Massey}  {\em J. L. Massey}, Linear codes with complementary duals, Discrete Math, {\bf} 106 (1992) 337-342.
\bibitem{Massey1}  {\em J. L. Massey}, Reversible codes, Inform. and Control, {\bf 7} (1964) 369-380.
\bibitem{Sendrier}  {\em N. Sendrier}, Linear codes with complementary duals meet the Gilbert--Varshamov bound, Discrete Math, {\bf 285} (2004) 345-347.
\bibitem{Sok}  {\em L. Sok, M. Shi and P. Sol\'{e}}, Construction of optimal LCD codes over large finite fields, Finite Fields Th App, {\bf 50} (2018) 138-153.
\bibitem{Yang}  {\em X. Yang, and J. L. Massey}, The condition for a cyclic code to have a complementary dual, Discrete Math, {\bf 126} (1994) 391-393.
\bibitem{Yang1}  {\em Y. Yang and W. Cai}, On self-dual constacyclic codes over finite fields, Des. Codes Cryptogr, {\bf 74} (2015) 355-364.
\bibitem{Zhu}  {\em B. Pang, S. Zhu and Z. Sun}, On LCD negacyclic codes over finite fields, Syst Sci Complex {\bf 31} (2018) 1065-1077.
\end{thebibliography}
\end{document}